\documentclass{llncs}
\usepackage{color,graphicx}

\begin{document}

\title{KATKA: A KRAKEN-like tool\\with $k$ given at query time}
\author{Travis Gagie\inst{1} \and Sana Kashgouli\inst{1} \and Ben Langmead\inst{2}}
\institute{Dalhousie University \and Johns Hopkins University}
\maketitle

\begin{abstract}
We describe a new tool, KATKA, that stores a phylogenetic tree $T$ such that later, given a pattern $P [1..m]$ and an integer $k$, it can quickly return the root of the smallest subtree of $T$ containing all the genomes in which the $k$-mer $P [i..i + k - 1]$ occurs, for $1 \leq i \leq m - k + 1$.  This is similar to KRAKEN's functionality but with $k$ given at query time instead of at construction time.
\end{abstract}

\section{Introduction}

KRAKEN~\cite{WS14,WLL19} is a popular tool that addresses the basic problem of determining where a fragment of DNA occurs in the Tree of Life, which arises for every sequencing read in a metagenomic dataset.  KRAKEN takes a phylogenetic tree $T$ and an integer $k$ and stores $T$ such that later, given a pattern $P [1..m]$, it can quickly return the root of the smallest subtree of $T$ containing all the genomes in which the $k$-mer $P [i..i + k - 1]$ occurs, for $1 \leq i \leq m - k + 1$.  For example, if $T$ is the small phylogenetic tree shown in Figure~\ref{fig:tree}, $k = 3$, and $P = \mathtt{TAGACA}$, then KRAKEN returns
\begin{itemize}
\item 8 for {\tt TAG} (which occurs in {\tt GATTAGAT} and {\tt GATTAGATA}),
\item 6 for {\tt AGA} (which occurs in {\tt AGATACAT}, {\tt GATTAGAT} and {\tt GATTAGATA}),
\item NULL for {\tt GAC} (which does not occur in $T$),
\item 2 for {\tt ACA} (which occurs in {\tt GATTACAT}, {\tt AGATACAT} and {\tt GATACAT}).
\end{itemize}
Notice that not all the genomes in the subtree returned for $P [i..i + k]$ need contain it: {\tt AGA} does not occur in {\tt GATTACAT} or {\tt GATACAT}.

\begin{figure}[t]
\begin{center}
\includegraphics[width=.35\textwidth]{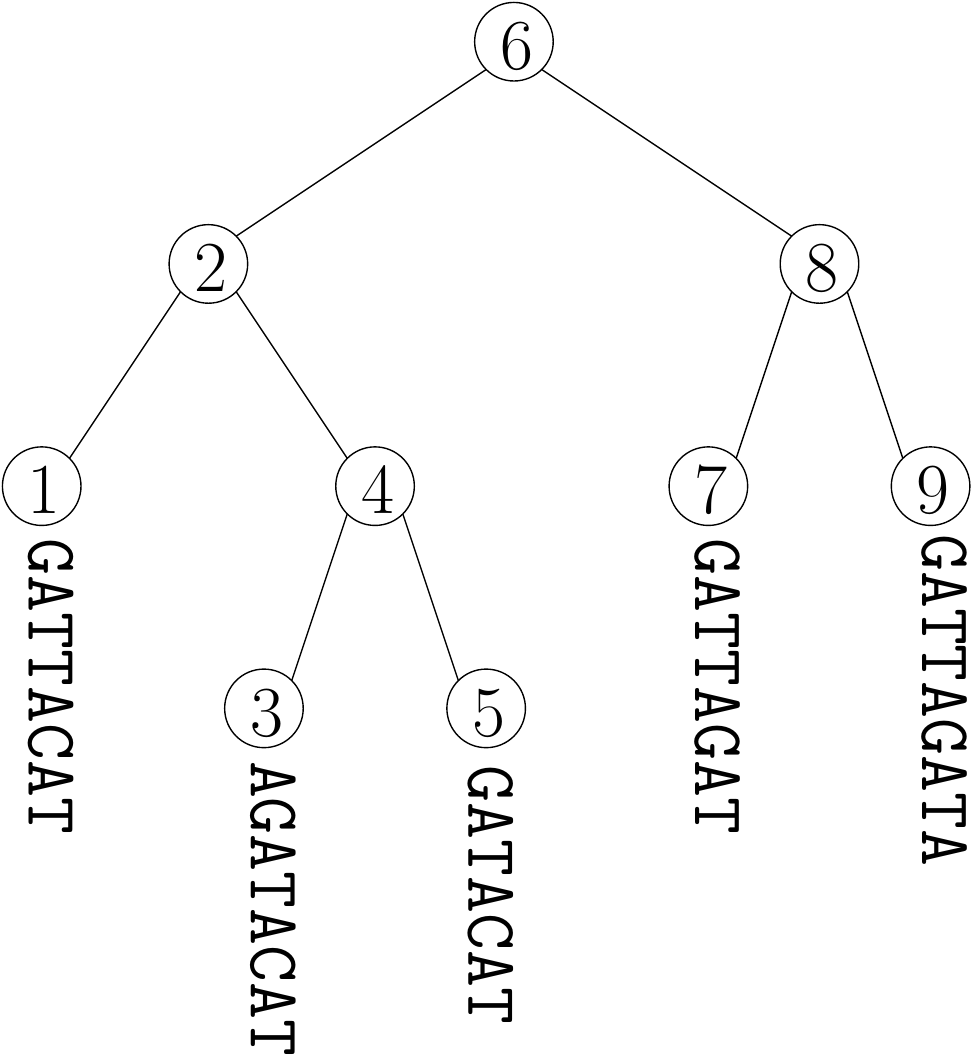}
\caption{A small phylogenetic tree.}
\label{fig:tree}
\end{center}
\end{figure}

KRAKEN is widely used in metagenomic analyses, especially taxonomic classification, but there are some applications for which we would rather give $k$ at query time instead of at construction time.  For example, Nasko et al.~\cite{NKPT18} showed that ``the [reference] database composition strongly influence[s] the performance'', with larger $k$ values generally working better as the database grows.  When the representation of strains or species in the database is skewed, therefore, it may be hard to choose a single $k$ that works well for all of them.  In this paper we describe a new tool, KATKA, that allows $k$ to be chosen at query time.  We are still optimizing, testing and extending KATKA and will report experimental results in a future paper.

\section{Design}

To simplify our presentation, in this paper we assume $T$ is binary (although our approach generalizes to higher-degree trees).  KATKA consists of three main components:
\begin{itemize}
\item a modified LZ77-index~\cite{KN13} for the concatenation of the genomes in $T$, in the order they appear from left to right in $T$ and separated by copies of a special character {\tt \$};
\item a modified LZ77-index for the reverse of that concatenation;
\item a lowest common ancestor (LCA) data structure for $T$.
\end{itemize}

Given $P [1..m]$ and $k$, we use the first and second indexes to find the leftmost and rightmost genomes in $T$, respectively, that contain the $k$-mer $P [i..i + k - 1]$, for $1 \leq i \leq m - k + 1$; we then use the LCA data structure to find the lowest common ancestor of those two genomes.  Since the two indexes are symmetric and the LCA data structure takes only about $2$ bits per vertex in $T$ and has constant query time, we describe only the first index.

To build the index for the concatenation, we compute its LZ77 parse and consider the phrases and co-lexicographically sort the set of their maximal non-empty suffixes not containing {\tt \$}, and consider the suffixes of the concatenation starting at phrase boundaries and lexicographically sort the set of their maximal prefixes not containing {\tt \$} (including the empty string $\varepsilon$ after the last phrase boundary).  We discard any of those maximal prefixes that do not occur starting at a phrase boundary immediately preceded by one of those maximal suffixes.

For our example, if the concatenation is
\[\mathtt{GATTACAT\$AGATACAT\$GATACAT\$GATTAGAT\$GATTAGATA}\,,\]
then its LZ77 parse is
\[\mathtt{G\ A\ T\ TA\ C\ AT\$\ AG\ ATA\ CAT\$G\ ATACAT\$GATT\ AGAT\$\ GATTAGATA}\,,\]
the co-lexicographically sorted set of maximal suffixes is
\[\mathtt{A}, \mathtt{TA}, \mathtt{ATA}, \mathtt{GATTAGATA}, \mathtt{C}, \mathtt{G}, \mathtt{AG}, \mathtt{T}, \mathtt{GATT}\,,\]
and the lexicographically sorted set of maximal prefixes is
\[\begin{array}{c}
\varepsilon, \mathtt{AGAT}, \mathtt{AGATACAT}, \mathtt{AT}, \mathtt{ATACAT}, \mathtt{ATTACAT}, \mathtt{CAT},\\
\mathtt{GATTACAT}, \mathtt{GATTAGATA}, \mathtt{TACAT}, \mathtt{TTACAT},
\end{array}\]
but we discard {\tt GATTACAT}, {\tt AGATACAT} and {\tt GATTAGATA} because they do not occur starting at a phrase boundary immediately preceded by one of the maximal suffixes.

We build a grid with the number $\ell$ at position $(x, y)$ if the genome at the $\ell$th vertex from the left in $T$ is the first one in which there is a phrase boundary immediately preceded by the co-lexicographically $x$th of the maximal suffixes and immediately followed by the lexicographically $y$th of the maximal prefixes.  Notice this grid will be of size at most $z \times z$ with at most $z$ numbers on it, where $z$ is the number of phrases in the LZ77 parse of the concatenation.  Figure~\ref{fig:grid} shows the grid for our example.

\begin{figure}[t]
\begin{center}
\begin{tabular}{|@{\hspace{1ex}}ccccccccc@{\hspace{1ex}}|l}
\rotatebox{-90}{\tt \textcolor{white}{--------}A\ } &
\rotatebox{-90}{\tt \textcolor{white}{-------}TA\ } &
\rotatebox{-90}{\tt \textcolor{white}{------}ATA\ } &
\rotatebox{-90}{\tt \textcolor{white}{}GATTAGATA\ } &
\rotatebox{-90}{\tt \textcolor{white}{--------}C\ } &
\rotatebox{-90}{\tt \textcolor{white}{--------}\textcolor{red}{G}\ } &
\rotatebox{-90}{\tt \textcolor{white}{-------}A\textcolor{red}{G}\ } &
\rotatebox{-90}{\tt \textcolor{white}{--------}T\ } &
\rotatebox{-90}{\tt \textcolor{white}{-----}GATT\ } \\
\hline
&&&9&&&&&&    \ $\varepsilon$\\
&&&&&&&&7& \tt\ AGAT\\
&&&&1&&&&& \tt\ \textcolor{red}{AT}\\
&&&&&\textcolor{red}{5}&\textcolor{red}{3}&&& \tt\ \textcolor{red}{AT}ACAT\\
&&&&&\textcolor{red}{1}&&&& \tt\ \textcolor{red}{AT}TACAT\\
&1&3&&&&&&& \tt\ CAT\\
&&&&&&&1&& \tt\ TACAT\\
1&&&&&&&&& \tt\ TTACAT\\
\hline
\end{tabular}
\caption{The grid we build for the concatenation in our example.}
\label{fig:grid}
\end{center}
\end{figure}

We store data structures such that given strings $\alpha$ and $\beta$, we can quickly find the minimum number in the box $[x_1, x_2] \times [y_1, y_2]$ on the grid, where $[x_1, x_2]$ is the co-lexicographic range of the maximal suffixes ending with $\alpha$ and $[y_1, y_2]$ is the lexicographic range of the maximal prefixes starting with $\beta$.  (For the index for the reversed concatenation, we find the maximum in the query box.)  In our example, if $\alpha = \mathtt{G}$ and $\beta = \mathtt{AT}$, then we should find 1.

For example, we can store Patricia trees for the compact tries for the reversed maximal suffixes and the maximal prefixes, together with a data structure supporting fast sequential access to the concatenation starting at any phrase boundary.  In the literature (see~\cite{Nav16} and references therein), the latter is usually an augmented straight-line program (SLP) for the concatenation; if the genomes in $T$ are similar enough, however, then in practice it could probably be simply a VCF file.  (We note that we can reuse the access data structure for the index for the reversed concatenation, augmented to support fast sequential access also at phrase boundaries in the reverse of the concatenation.)  Figure~\ref{fig:tries} shows the compact tries for our example, with each black leaf indicating that one of the strings in the set ended at the parent of that leaf.

\begin{figure}[t]
\begin{center}
\includegraphics[width=.9\textwidth]{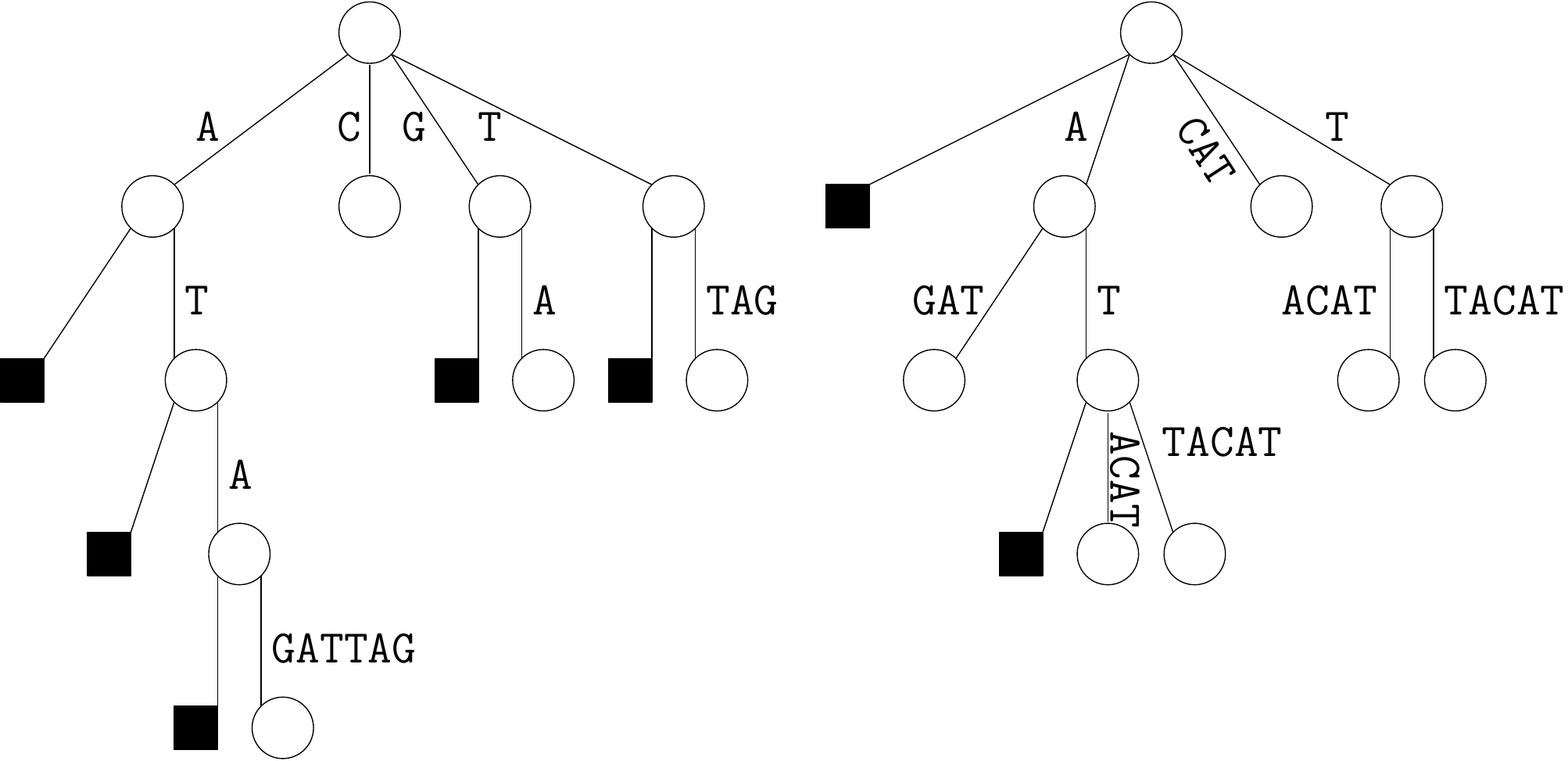}
\caption{The compact tries for the concatenation in our example.}
\label{fig:tries}
\end{center}
\end{figure}

Nekrich~\cite{Nek21} recently showed how to store the grid in $O (z)$ space and support 2-dimensional range-minimum queries in $O (\log^{\epsilon} z)$ time, for any constant $\epsilon > 0$.  For simplicity, we consider his data structure in our analysis even though we are not aware of any implementation yet.

\section{Queries}

Given a pattern $P [1..m]$ and an integer $k$, for every substring $P [i..j]$ of $P$ with length at most $k$, we find and verify the locus for the reverse of $P [i..j]$ in the compact trie for the reversed maximal suffixes, and the locus for $P [i..j]$ in the compact trie for the maximal prefixes.  (Patricia trees can return false positives when the sought pattern does not occur, so we must verify the loci by, for example, extracting their path labels from the SLP.)

By combining the searches for $P [i], P [i..i + 1], \ldots, P [i..i + k - 1]$, we make a total of $O (m)$ descents in the Patricia trees, each to a string-depth of at most $k$; we extract $O (m)$ substrings from the concatenation, each of length at most $k$ and starting at a phrase boundary, to verify the loci.  With care, this takes a total of $O (k m)$ time in the worst case.  When searching standard LZ77-indexes in practice, however, ``queries often die in the Patricia trees''~\cite{Nav??} --- because of mismatches between characters in the pattern and the first characters in edge labels --- which speeds up queries.

For each $k$-mer $P [i..i + k - 1]$ in $P$ and each way to split $P [i..i + k - 1]$ into a non-empty prefix $P [i..j]$ and a suffix $P [j + 1..i + k - 1]$, we use a 2-dimensional range-minimum query to find the minimum number in the box for $\alpha = P [i..j]$ and $\beta = P [j + 1..i + k - 1]$ in $O (\log^\epsilon z)$ time.

By the definition of the LZ77 parse, the first occurrence of $P [i..i + k - 1]$ in the concatenation crosses or ends at a phrase boundary.  It follows that, by taking the minimum of the minima we find, in $O (k \log^\epsilon z)$ time we find the leftmost genome in $T$ that contains $P [i..i + k - 1]$.  Repeating this for every value of $i$ takes $O (k m \log^\epsilon z)$ time.

By storing symmetric data structures for the reverse of the concatenation and querying them, we can find the rightmost genome in $T$ that contains $P [i..i + k - 1]$, for $1 \leq i \leq m - k + 1$.  With the LCA data structure for $T$, we can find the lowest common ancestor of the two genomes, which is the root of the smallest subtree of $T$ containing all the genomes in which the $k$-mer $P [i..i + k - 1]$ occurs.

For our example, if $P = \mathtt{TAGACA}$ and $k = 3$, then we find and verify the loci for
\[\mathtt{T}, \mathtt{A}, \mathtt{AT}, \mathtt{G}, \mathtt{GA}, \mathtt{GAT}, \mathtt{A}, \mathtt{AG}, \mathtt{AGA}, \mathtt{C}, \mathtt{CA}, \mathtt{CAG}, \mathtt{A}, \mathtt{AC}, \mathtt{ACA}\]
in the compact trie for the reversed maximal suffixes, and the loci for
\[\mathtt{A}, \mathtt{AG}, \mathtt{AGA}, \mathtt{G}, \mathtt{GA}, \mathtt{GAC}, \mathtt{A}, \mathtt{AC}, \mathtt{ACA}, \mathtt{C}, \mathtt{CA}, \mathtt{A}\]
in the compact trie for the maximal prefixes.

For $P [1..3] = \mathtt{TAG}$, we look up the minimum number 7 in the box for $\alpha = \mathtt{T}$ and the locus $\beta = \mathtt{AGAT}$ for {\tt AG}; since {\tt G} has no locus in the compact trie for the maximal prefixes and {\tt GAT} has no locus in the compact trie for the maximal reversed suffixes, we correctly conclude that the leftmost genome in $T$ containing {\tt TAG} is at vertex 7.  A symmetric process with the index for the reversed concatenation tells us the rightmost genome in $T$ containing {\tt TAG} is at vertex 9, and then an LCA query tells us that vertex 8 is the root of the smallest subtree containing all the genomes in which {\tt TAG} occurs.

\begin{theorem}
Given a phylogenetic tree $T$ whose $g$ genomes have total length $n$, we can store $T$ in $O (z \log n + g / \log n)$ space, where $z$ is the number of phrases in the LZ77 parse of the concatenation of the genomes in $T$ (separated by copies of a special character), such that when given a pattern $P [1..m]$ and an integer $k$, for $1 \leq i \leq m - k + 1$ we can find the root of the smallest subtree of $T$ containing all genomes in which the $k$-mer $P [i..i + k - 1]$ of $P$ occurs, in $O (k m \log^\epsilon z)$ total time.
\end{theorem}

\begin{proof}
The LCA data structure takes $2 g + o (g)$ bits, which is $O (g / \log n)$ words (assuming $\Omega (\log n)$-bit words).  An SLP for the concatenation with bookmarks permitting sequential access with constant overhead from the phrase boundaries in the parses of the concatenation and its reverse, takes $O (z \log n)$ space.  For the concatenation, the Patricia trees and the instance of Nekrich's 2-dimensional range-minimum data structure take $O (z)$ space; for the reverse of the concatenation, they take space proportional to the number of phrases in its LZ77 parse, which is $O (z \log n)$.  In total, we use $O (z \log n + g / \log n)$ space.  As we have described, we make $O (m)$ descents in the Patricia trees, each to string-depth at most $k$, and extract only $O (m)$ substrings, each of length at most $k$, from the concatenation and its reverse.  The time is dominated by the $O (k m)$ range-minimum queries, which take $O (\log^\epsilon z)$ time each.
\end{proof}

\section{Future Work}

In addition to optimizing and testing KATKA, we are also investigating adapting it to work with maximal exact matches (MEMs) instead of $k$-mers.  For example, if we store $O (z)$-space z-fast tries~\cite{BBPV10} for the Patricia trees then, for each way to split $P$ into a non-empty prefix $P [1..i]$ and a suffix $P [i + 1..m]$, we can find the loci of $P [1..i]$ reversed and $P [i + 1..m]$ in $O (\log m)$ time.  We can verify those loci in $O (\log n)$ time by augmenting the SLP to return fingerprints, without changing its $O (z \log n)$ space bound~\cite{BGHSVV17}.

With an $O (z)$-space data structure supporting heaviest induced ancestor queries in $O \left( \frac{\log^2 z}{\log \log z} \right)$ time~\cite{AHGT22,GGN13,Gao22}, in that time we can find the longest substring $P [h..j]$ with $h \leq i \leq j$ that occurs in $T$ with $P [h..i]$ immediately to the left of a phrase boundary and $P [h + 1..j]$ immediately to its right.  Note $P [h..j]$ must be a MEM.  With 2-dimensional range-minimum and range-maximum queries, we can find the indexes of the leftmost and rightmost genomes in which $P [h..i]$ occurs immediately to the left of a phrase boundary and $P [h + 1..j]$ immediately to its right.  We still use a total of $O (z \log n + g / \log n)$ space and now we use a total of $O \left( m \left( \frac{\log^2 z}{\log \log n} + \log n \right) \right)$ time.

Unfortunately, we may not find every MEM this way: it may be that, for some MEM $P [h..j]$ and every $i$ between $h$ and $j$, either $P [h..j]$ is not split into $P [h..i]$ and $P [i + 1..j]$ by any phrase boundary or some longer MEM is split into $P [h'..i]$ and $P [i + 1..j']$ by a phrase boundary.  For any MEM we do not find, however, we do find another MEM at least as long that overlaps it.

A more serious drawback to this scheme is that it is probably quite impractical (for example, we are not aware of any implementation of a data structure supporting fast heaviest induced ancestor queries, either).  Fortunately, there is probably a simple and practical solution --- although possibly without good worst-case bounds --- since the problems we are considering are similar to those covered in Subsection 7.1 of Navarro's~\cite{Nav14} survey on wavelet trees.

Finally, we are investigating adapting results~\cite{Nav19} using LZ77-indexes for document-listing, in order to find the number of genomes in $T$ in which each $k$-mer of $P$ occurs.  It is easy to store a small data structure that reports the number of genomes in the smallest subtree for a $k$-mer, so we may be able to determine what fraction contain that $k$-mer.

\subsection*{Acknowledgments}

Many thanks to Nathaniel Brown, Younan Gao, Simon Gog, Meng He, Finlay Maguire and Gonzalo Navarro, for helpful discussions.

\end{document}